\documentclass[]{article}

\usepackage{authblk}

\usepackage{amssymb,amsthm}
\usepackage{graphicx}
\usepackage[PostScript=dvips]{diagrams}
\usepackage{verbatim}

\usepackage{algorithm}
\usepackage{algpseudocode}

\newtheorem{definition}{Definition}
\newtheorem{remark}{Remark}
\newtheorem{proposition}{Proposition}
\newtheorem{example}{Example}

\usepackage{color}

\newcommand{\ubar}[1]{\text{\b{$#1$}}}

\usepackage{mathtools}
\newlength{\arrow}
\begin{document}

\title{Scheduling of Event-Triggered Networked Control Systems using Timed Game Automata}


\author[1]{Dieky Adzkiya}
\author[2]{Manuel Mazo, Jr.}
\affil[1]{Department of Mathematics, Institut Teknologi Sepuluh Nopember, Indonesia. Part of this work was done in the Delft Center for Systems and Control, TU Delft, The Netherlands. (e-mail: dieky@matematika.its.ac.id)}
\affil[2]{Delft Center for Systems and Control, TU Delft - Delft University of Technology, The Netherlands. (e-mail: m.mazo@tudelft.nl).}
\date{}                     
\setcounter{Maxaffil}{0}


\maketitle

\begin{abstract}
We discuss the scheduling of a set of networked control systems implemented over a shared communication network.
Each control loop is described by a linear-time-invariant (LTI) system with an event-triggered implementation.
We assume the network can be used by at most one control loop at any time instant and after each controller update, a pre-defined channel occupancy time elapses before the network is available.
In our framework we offer
the scheduler two options to avoid conflicts:
using the event-triggering mechanism, where the scheduler can choose the triggering coefficient; or
forcing controller updates at an earlier pre-defined time.
Our objective is avoiding communication conflict while guaranteeing stability of all control loops.
We
formulate the original scheduling problem as a control synthesis problem over a network of timed game automata (NTGA) with a safety objective.
The NTGA is obtained by taking the parallel composition of the timed game automata (TGA) associated with the network and with all control loops.
The construction of TGA associated with control loops leverages recent results on the abstraction of timing models of event-triggered LTI systems.
In our problem, the safety objective is to avoid that update requests from a control loop happen while the network is in use by another task.
We showcase the results in some examples.
\end{abstract}


\section{Introduction}

Networked control systems (NCSs) are spatially distributed systems in which the communication between sensors, actuators and controllers occurs through a shared band-limited digital communication network, as shown in Fig.\ \ref{fig:ncs}.
Such structures bring many advantages, for instance reduced wiring and maintenance costs as well as an increased flexibility and reconfigurability.
NCSs occur in numerous applications, including power systems \cite{Liu2009}, aircrafts and automobiles \cite{Thompson2004}, and process control \cite{Lehmann2011}.

\begin{figure*}[!t]
\centering
\includegraphics[width=\textwidth]{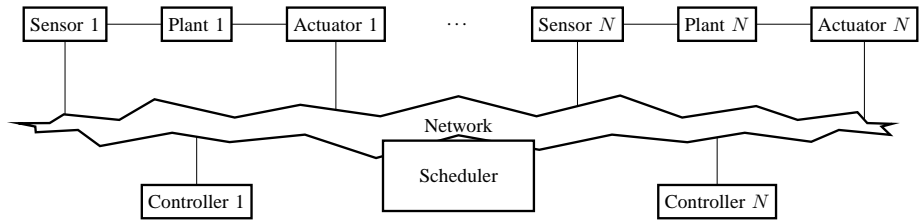}
\caption{Topology of a set of $N$ networked control systems.}\label{fig:ncs}
\end{figure*}

Notice that NCSs are implemented over shared communication resources, most often over digital channels.
The impact of such communication infrastructures on control systems has been studied in the last decade \cite{special_proceedings, reviewCUDR, Hespanha07asurvey, HandbookNECS}.
In particular the applicability of wireless communications in
NCSs has been discussed in \cite{RUNES, Rabi08, MazoTACWSAN} among others.
The delays introduced by these shared resources on the feedback loops are critical to guarantee stability and performance.
Furthermore, when several control loops are implemented over a shared communication channel, bandwidth becomes a scarce resource, the usage of which needs to be minimized by each controller.


For these reasons,
the traditional time-triggered controller implementations, i.e.\ based on periodic sampling, are not suitable anymore.
With the objective of minimizing the bandwidth usage,
event-based approaches resulting in aperiodic controller updates have been proposed in \cite{Astrom, Tabuada, Velasco, Heemels08, AntaTAC, MazoAnta}.
These aperiodic paradigms introduce a new challenge in the design: the scheduling of transmissions.

A communication network has finitely many channels. The number of channels represents the maximum number of messages that can be sent simultaneously over the network. If the number of channels equals the number of control loops, we do not need any scheduler because every control loop has its own communication channel. However in practice, the number of channels is smaller than the number of control loops. Thus we need a scheduler to decide which control loop has access to the network at any time instant while guaranteeing stability of all control loops. Additionally the scheduler can optimize a certain combination of control performance and bandwidth usage.
In the context of periodic traffic sources, or aperiodic with known deadlines a-priori, there are well-studied scheduling techniques that are capable of guaranteeing certain delay bounds for the traffic \cite{sprunt1989aperiodic, buttazzo2011hard, WalshYe2001}. In the event-based context, i.e.\ sporadic traffic with unknown deadlines, the problem has been less studied and becomes more challenging \cite{Cervin08,AlAreqi2013,AlAreqi2014,Reimann2013}.
In \cite{AlAreqi2013,AlAreqi2014,Reimann2013}, the authors propose a joint design (codesign) of a control law and a scheduling law for several types of NCSs.
Although the co-design strategy can improve the control performance significantly,
if a new control loop is introduced to the NCSs,
the whole co-design procedure has to be executed again, which can be extremely time consuming.
In order to mitigate this issue, we propose an approach that separates the design of controllers and schedulers.

In this paper, we design a scheduler for a set of NCSs over a shared communication network (cf.\ Fig.\ \ref{fig:ncs}).
Our objective is avoiding communication conflict while guaranteeing stability of all NCSs.
Each NCS is a linear-time-invariant (LTI) system. The controllers are implemented in an event-triggered fashion where the delays between reading the state and updating the actuators are ignored \cite{Tabuada}. With respect to the shared communication network, we assume it has a single communication channel. Furthermore after each controller update, a pre-defined channel occupancy time elapses before the network is available. We consider schedulers that after each transmission of measurements, decide the policy for the next update. These policies can be to either let the next update be decided based on a triggering mechanism (to be chosen among a set of them guaranteeing different performances) or forced to be at an earlier pre-defined time.
If we do not allow the scheduler to force earlier controller updates, the bandwidth usage is decreased at the cost of worse control performance (or slower convergence).
On the other hand if we allow the scheduler to force earlier controller updates, we obtain a better control performance (or faster convergence) at the cost of increased bandwidth usage.
In this case, nothing prevents the scheduler from always using earlier updates and never use the event-triggering mechanism.  This may result in an undesired over-use of the communication channel, and could be prevented by introducing costs to the model, which would result in priced timed game automata (PTGA). Unfortunately to the best of our knowledge, no results are available in the literature allowing the synthesis of strategies over PTGA with safety objectives.
Only the synthesis of strategies over PTGA with a reachability objective are available \cite{Bouyer2005,Bouyer2005a}.
Thus we propose an alternative approach to prevent the undesired schedule by limiting the consecutive earlier updates.

The scheduling problem can be formulated as a timed safety game:
given a model and a set of bad states,
we seek to construct a strategy such that the model supervised by the strategy constantly avoids the bad states.
In our problem, the safety problem at hand is to avoid that update requests from a control loop happen while the network is in use by another task.
We focus our attention to the design of a scheduler by leveraging techniques
originally developed for network of timed game automata (NTGA) \cite{Maler1995}.
An NTGA is the parallel composition of timed game automata (TGA), which are
timed safety automata (TSA) in with the set of actions is partitioned into controllable and uncontrollable actions.
We choose NTGA modeling framework because of the following two reasons.
First of all, it allows us to extend the methods in \cite{ArmanTCoN}. The authors of \cite{ArmanTCoN} discuss formal abstraction of the timing behavior of LTI systems with event-triggered implementation as TSA.
The second reason is that
the solution of timed safety game over NTGA can be computed by using backward algorithms \cite{DeAlfaro2001,Maler1995} or on-the-fly algorithms \cite{Cassez2005}.
Moreover the algorithms have been implemented in some freely available software tools \cite{Asarin1998,Cassez2005}.
The procedure to generate the scheduler (or the strategy) is as follows. First, we construct an NTGA from a set of NCSs.
The NTGA is obtained by taking the parallel composition of the TGA associated with the network and with all control loops.
Then we characterize the bad states, i.e.\ states corresponding to a communication conflict. Finally the scheduler is defined as the solution of timed safety game over the NTGA.

%
Timed automata (TA) \cite{Alur1994} are a general modeling framework for a wide range of real-time systems, such as in web services \cite{Ravn2011}, audio/video protocols \cite{Havelund1997}, bounded retransmission protocols \cite{DArgenio1997}, collision avoidance protocols \cite{Aceto1998,Jensen1996} and commercial field bus protocols \cite{David2000}.
Timed safety automata (TSA) \cite{Henzinger1994} are a simplified version of TA.
In order to enforce progress properties, TSA use local invariant conditions whereas TA use B{\"u}chi or Muller accepting conditions.
%
The scheduling problem over TA and its variants has been already studied in the literature e.g.\ applied to the scheduling of a steel plant, a job shop and a task graph in \cite{Fehnker1999}, \cite{Abdeddaim2001}, and \cite{Abdeddaim2003}, respectively.
Furthermore the optimal scheduling of a production w.r.t.\ a predefined cost for a finite time horizon has been investigated in \cite{Behrmann2005a,Behrmann2005b}.
In this case, the models are TA with weights (or costs) on both locations and edges, so called priced timed automata in \cite{Behrmann2001} and weighted timed automata in \cite{Alur2001}. Finally the optimal scheduling for infinite time horizon is discussed in \cite{Bouyer2008}.

The rest of this manuscript is structured as follows.
Section \ref{sec:modelprel} recalls some modeling frameworks and preliminaries.
Section \ref{sec:scheduling} proposes a procedure to synthesize a conflict-free scheduler.
Two experimental results are discussed in Section \ref{sec:case}.
The first experiment allows the scheduler to force earlier controller updates where the number of consecutive earlier updates is limited to 4.
The second experiment does not allow the scheduler to force earlier controller updates. In this experiment, the scheduler can choose the triggering coefficient among three choices.
Finally the conclusions and possible future research directions are summarized in Section \ref{sec:concl}.

\section{Models and Preliminaries}
\label{sec:modelprel}

\subsection{Timed Automata}
\label{sec:modelprel-ta}


A timed automaton (TA) \cite{Alur1994} is a finite automaton (namely, a directed graph containing finitely many nodes and finitely many labeled edges)
extended with real-valued variables, which is usually employed to model real-time systems.
The real-valued variables model logical clocks, that are initialized to zero when the system is started and thereafter increase synchronously at the same rate.
We shall refer to these variables as simply ``clocks".
Clock constraints are used to restrict the behavior of the automaton.
An edge transition can be taken when the edge is enabled. Edges are enabled if the values of the clocks satisfy the guard conditions associated with the edge.
Additionally, some clocks may be reset to zero when an edge is taken.
Originally, B{\"u}chi and Muller accepting conditions are used to enforce progress properties \cite{Alur1994}.
A simplified version called timed safety automata \cite{Henzinger1994} uses local invariant conditions to specify progress properties.
In this work, we focus on timed safety automata and refer them as timed automata for simplicity.

We define $C$ as the set of finitely many clocks,
$\mathsf{Act}$ as the set of finitely many actions and
$\mathbb N_0$ as the set of natural numbers including zero $\{0,1,\dots\}$.
A clock constraint is a conjunctive formula of atomic constraints of the form $x \bowtie n$ or $x - y \bowtie n$ for $x,y \in  C$, $\bowtie \in \{\leq,<,=,>,\geq\}$ and $n \in \mathbb N_0$.
Clock constraints will be used as guards on edges and location invariants.
We use $\mathcal B(C)$ to denote the set of clock constraints.
\begin{definition}[Timed Automaton]\label{def:ta}
A timed automaton $\mathsf{TA}$ is a sextuple\\$(L,\ell_0,\mathsf{Act},C,E,\mathsf{Inv})$ where
\begin{itemize}
\item $L$ is a set of finitely many locations (or nodes);
\item $\ell_0 \in L$ is an initial location;
\item $\mathsf{Act}$ is a set of finitely many actions;
\item $C$ is a set of finitely many real-valued clocks;
\item $E \subseteq L \times \mathcal B(C) \times \mathsf{Act} \times 2^{C} \times L$ is a set of edges;
\item $\mathsf{Inv} : L \to \mathcal B(C)$ assigns invariants to locations.\footnote{Recall that $2^C$ denotes the power set of $C$.}
\end{itemize}
Location invariants are restricted to constraints that are downwards closed, in the form: $c \leq n$ or $c < n$ where $c$ is a clock and $n \in \mathbb N_0$.
\end{definition}
Sometimes we write $\ell \rTo^{g,a,\mathbf r} \ell'$ 
when $(\ell,g,a,\mathbf r,\ell') \in E$.
Furthermore we write $\ell \rTo \ell'$ 
to denote the existence of an edge from $\ell$ to $\ell'$ with arbitrary labels.

The semantics of a TA are defined as a transition system where a state consists of the current location and the current value of clocks.
There are two types of transitions between states depending on whether the automaton: delays for some time (a delay transition), or takes an enabled edge (a discrete transition).

To keep track of clock values, we use functions known as clock assignments
$u:C \to \mathbb R_{\geq 0}$ and we employ $u \vDash g$ ($u$ satisfies $g$) to denote that the clock values of $u$ satisfy the guard $g$.
For $d \in \mathbb R_{\geq 0}$, let $u + d$ denote the clock assignment that maps all $c \in C$ to $u(c) + d$.
For a set of clocks $\mathbf c \subseteq C$, let $u[\mathbf c]$ denote the clock assignment that maps all clocks in $\mathbf c$ to 0 and agrees with $u$ for the rest of clocks in $C \setminus \mathbf c$.
\begin{definition}[Operational Semantics]
The semantics of a timed automaton is a transition system (also known as a timed transition system) in which states are pairs of location $\ell$ and clock assignment $u$, and
transitions are defined by the rules:
\begin{itemize}
\item Delay transition: $(\ell,u) \rTo^{d}_{\mathit{TS}} (\ell, u+d)$ if $u \vDash \mathsf{Inv}(\ell)$ and $(u+d) \vDash \mathsf{Inv}(\ell)$ for a non-negative real number $d \in \mathbb{R}_{\geq 0}$;
\item Discrete transition: $(\ell,u) \rTo^{a}_{\mathit{TS}} (\ell',u')$ if $\ell \rTo^{g,a,\mathbf r} \ell'$, $u \vDash g$, $u' = u[\mathbf r]$ and $u' \vDash \mathsf{Inv}(\ell')$.
\end{itemize}
A run of a timed automaton is a sequence of alternating delay and discrete transitions in the transition system.
\end{definition}
We denote by $\mathsf{Runs}(\mathsf{TA})$ the set of runs of timed automaton $\mathsf{TA}$ starting from the initial state $(\ell_0,u_0)$ where $u_0$ is a clock assignment that maps all $c \in C$ to 0. Additionally, if $\rho$ is a finite run, the last state of the run is denoted by $\mathit{last}(\rho)$.

The set of actions \textsf{Act} (cf.\ Definition \ref{def:ta})
is assumed to consists of symbols for input actions $a?$, output actions $a!$ and internal actions $\ast$.
Synchronous communication between different TA is done by hand-shake synchronization using input and output actions.

To model concurrent systems, several TAs can be extended with parallel composition that takes into account the synchronous communication.
Parallel composition of TAs is also called network of timed automata (NTA).
Essentially the parallel composition of a set of TAs is the product of the TAs.
Building the product timed automaton is an entirely syntactical but computationally expensive operation.
The reader is referred to \cite[Sec.\ 5]{Bengtsson2004} for an example on the composition of two TAs.

The semantics of an NTA are defined as a transition system where a state consists of a vector of current locations and the current value of clocks in all TAs \cite{uppaal2004}.


\subsection{Timed Game Automata}
\label{sec:modelprel-tga}


A timed game automaton is a timed automaton in which the set of actions is partitioned into controllable and uncontrollable actions.
The former are actions that can be triggered by the controller, whereas the latter only by the environment/opponent.

\begin{definition}[Timed Game Automaton]\label{def:tga}
A timed game automaton $\mathsf{TGA}$ is a septuple $(L,\ell_0,\mathsf{Act}_c,\mathsf{Act}_u,C,E,\mathsf{Inv})$ where
\begin{itemize}
\item $(L,\ell_0,\mathsf{Act}_c \cup \mathsf{Act}_u,C,E,\mathsf{Inv})$ is a timed automaton;
\item $\mathsf{Act}_c$ is a set of controllable actions;
\item $\mathsf{Act}_u$ is a set of uncontrollable actions;
\item $\mathsf{Act}_c \cap \mathsf{Act}_u = \emptyset$.
\end{itemize}
\end{definition}

Similar to TA, TGA can also be extended with parallel composition (essentially the synchronized cartesian product of TGA).
The parallel composition of TGAs is called a ``network of timed game automata" (NTGA) which is formally defined as: 


\begin{definition}[Parallel Composition]\label{def:tga-composition}
Let $\mathsf{TGA}^i = (L^i,\ell_0^i,$ $\mathsf{Act}_c^i,\mathsf{Act}_u^i,C^i,E^i,\mathsf{Inv}^i)$
be a timed game automaton
for $i \in \{1,\dots,n\}$.
The parallel composition of $\mathsf{TGA}_1$, $\dots$, $\mathsf{TGA}_n$ denoted by $\mathsf{TGA}_1 \mid \dots \mid \mathsf{TGA}_n$ is a timed game automaton $\mathsf{TGA} = (L,\ell_0,\mathsf{Act}_c,\mathsf{Act}_u,C,E,\mathsf{Inv})$ where
\begin{itemize}
\item $L = L^1 \times \dots \times L^n$;
\item $\ell_0 = (\ell_0^1,\dots,\ell_0^n)$;
\item $\mathsf{Act}_c = \{ \ast \} \cup \bigcup_{i=1}^n \{ a \in \mathsf{Act}_c^i \mid a \text{ is an internal action} \}$;
\item $\mathsf{Act}_u = \{ \circledast \}  \cup \bigcup_{i=1}^n \{ a \in \mathsf{Act}_u^i \mid a \text{ is an internal action} \}$;
\item $C = C^1 \cup \dots \cup C^n$;
\item $E$ is defined according to the following two rules:
\begin{itemize}
\item a TA makes a move on its own via its internal action: the edge is controllable iff the internal action is controllable;
\item two TAs move simultaneously via a synchronizing action: the edge is controllable iff both input and output actions are controllable (i.e.\ the environment has priority over the controller);
\end{itemize}
\item $\mathsf{Inv}((\ell_1,\dots,\ell_n)) = \mathsf{Inv}^1(\ell_1) \wedge \dots \wedge \mathsf{Inv}^n(\ell_n)$.
\end{itemize}
\end{definition}

In the parallel composition of TGAs, a pair of input and output actions is denoted as a single action.
Thus the sets $\mathsf{Act}_c$ and $\mathsf{Act}_u$ do not contain any input and output actions.
A synchronizing action should be defined as an element of $\mathsf{Act}_c$ if it is controllable and an element of $\mathsf{Act}_u$ if it is not controllable.
In Definition \ref{def:tga-composition}, let us remark that both $\mathsf{Act}_c$ and $\mathsf{Act}_u$ do not contain synchronizing actions for simplicity.
Any controllable synchronizing action is denoted by $\ast$, whereas any uncontrollable synchronizing action is denoted by $\circledast$.


Given an NTGA, we are interested in solving the following safety objective:
is it possible to find a strategy for the triggering of controllable actions guaranteeing that a set of pre-specified bad states are never reached regardless of what and when uncontrollable actions take place?
More formally given an NTGA and a set of bad states $\mathcal A$,
we seek to construct a strategy $f$ such that the NTGA supervised by $f$ constantly avoids $\mathcal A$.

A strategy \cite{Maler1995} is a function that during the course of a game constantly gives information about what the controller should do in order to win the game.
At any given situation, the strategy could suggest the controller to either ``take a particular controllable action'' or ``do nothing at this point in time'', i.e. delay, 
which will be denoted by the symbol (controllable action) $\lambda$.

\begin{definition}[Strategy~{\cite[Definition 3]{Cassez2005}}]\label{def:strategy}
Let $\mathsf{TGA} = (L,\ell_0,$ $\mathsf{Act}_c,\mathsf{Act}_u,C,E,\mathsf{Inv})$ be a timed game automaton.
We define $\mathsf{TA} = (L,\ell_0,\mathsf{Act}_c \cup \mathsf{Act}_u,C,E,\mathsf{Inv})$ as the timed automaton derived from the timed game automaton.
A strategy $f$ over $\mathsf{TGA}$ is a partial function from $\mathsf{Runs}(\mathsf{TA})$ to $\mathsf{Act}_c \cup \{ \lambda \}$ s.t.\ for every finite run $\rho$, if $f(\rho) \in \mathsf{Act}_c$ then $\mathit{last}(\rho) \rTo^{f(\rho)}_{\mathit{TS}} (\ell',u')$ for some $(\ell',u')$.
\end{definition}

A strategy $f$ over $\mathsf{TGA}$ is called state-based or memoryless whenever $\mathit{last}(\rho) = \mathit{last}(\rho')$ implies $f(\rho) = f(\rho')$, for each $\rho,\rho' \in \mathsf{Runs}(\mathsf{TA})$.
The restricted behavior of an NTGA controlled with some strategy $f$ is defined by the notion of outcome \cite{DeAlfaro2001}.

A strategy $f$ is winning from a state if all maximal runs \cite[p.\ 70]{Cassez2005} in the outcome originated from that state are winning.
A state is winning if there exists a winning strategy $f$ from that state.
The winning states can be computed by using backward algorithms \cite{DeAlfaro2001,Maler1995} or on-the-fly algorithms \cite{Cassez2005}.
Software tools are also available that solve safety control problems, e.g.\ the implementation from Verimag \cite{Asarin1998} or UPPAAL-Tiga \cite{Cassez2005}, which
implement the backward and on-the-fly algorithms respectively.

\subsection{Event Triggered Control Systems}
\label{sec:modelprel-etc}

We consider linear-time-invariant (LTI) systems of the form
\begin{equation}
\dot \xi(t) = A \xi(t) + B \upsilon(t), \quad \xi(t) \in \mathbb{R}^n, \quad \upsilon(t) \in \mathbb{R}^m
\label{eqn:lti-system}
\end{equation}
where $A$ and $B$ are matrices of appropriate dimensions.
We assume the existence of linear state-feedback laws $\upsilon(t) = K \xi(t)$ rendering the closed-loop system globally asymptotically stable, where $K$ is a matrix of appropriate dimensions.

Assume a sample-and-hold implementation of the control law is in place keeping
the input signal constant between update times, i.e.
\begin{equation}
\begin{array}{l}
\upsilon(t) = K \xi(t_k), \quad t \in [t_k,t_{k+1}[,
\end{array}
\label{eqn:lti-sample-hold}
\end{equation}
where $t_0,t_1,\dots$ is a divergent sequence of update times.
For simplicity of presentation, we ignore the presence of delays between reading the state and updating the actuators.
The interested reader is referred to \cite{Tabuada} for more details, including accounting for delays.

In event-triggered implementations, the sequence 
of update times is 
decided on run-time based on the state of the plant \cite{Tabuada}.
%
Let $\xi(t)$ represent the solution of \eqref{eqn:lti-system}-\eqref{eqn:lti-sample-hold}.
%
We define an auxiliary variable $e(t)$ representing the difference between the sampled state $\xi(t_k)$ and the current state $\xi(t)$ of the system:
\begin{displaymath}
\begin{array}{l}
e(t) = \xi(t_k) - \xi(t), \qquad t \in [t_k,t_{k+1}[, \quad k \in \mathbb{N}_0.
\end{array}
\end{displaymath}
The event-triggering approach in \cite{Tabuada},
proposes the following sampling triggering law:
\begin{equation}
\begin{array}{l}
t_{k+1} = \min \{t \mid t> t_k \text{ and } | e(t) |^2 \geq \sigma | \xi(t) |^2 \},
\label{eqn:tabtrig}
\end{array}
\end{equation}
where $\sigma \in ]0, \bar\sigma[ \subset \mathbb{R}^+$ is the triggering coefficient, which establishes a trade off between quality of control (convergence rate to the equilibrium) and the amount of transmissions triggered.
The inter-sample time of the state $x$, denoted by $\tau_\sigma(x)$, is defined as the time between consecutive updates when the sampled state is $x$:
\begin{equation}
\begin{array}{l}
\tau_\sigma(x) = \min \{t \mid | e(t) |^2 \geq \sigma | \xi(t) |^2 \text{ and } \xi(0)=x \}.
\label{eqn:tabtrig-1}
\end{array}
\end{equation}

\subsection{Abstraction of Event Triggered Control Systems as Timed Automata}
\label{sec:modelprel-abs}


In \cite{ArmanTCoN}, the authors propose an approach to characterize the sampling behavior of LTI systems with event-triggered implementation as TAs.
The approach abstracts the spatial and temporal dependencies of the original system.
The following definitions summarize the approach.

\begin{definition}[Flow Pipe]
The set of reachable states or the flow pipe at the time interval $[t_1,t_2]$ from a set of initial states $X_0$ is denoted by
\begin{displaymath}
\mathcal X_{[t_1,t_2]}(X_0) = \bigcup_{t \in [t_1,t_2]} \{ \xi(t) \mid \xi(0) \in X_0 \}.
\end{displaymath}
\end{definition}

\begin{definition}[\cite{ArmanTCoN}]\label{def:lti2ta-arman}
A timed automaton abstracting the triggering timing behavior of system \eqref{eqn:lti-system}-\eqref{eqn:lti-sample-hold} with triggering coefficient $\sigma$ is given by\\ $\mathsf{TA}^{\sigma} = (L^{\sigma},\ell_0^{\sigma},\mathsf{Act}^{\sigma},C^{\sigma},E^{\sigma},\mathsf{Inv}^{\sigma})$ where
\begin{itemize}
\item $L^{\sigma} = \{ \mathcal R_1^{\sigma},\dots,\mathcal R_q^{\sigma} \}$;
\item $\ell_0^{\sigma} = \mathcal R_s^{\sigma}$ such that $\xi(0) \in \mathcal R_s^{\sigma}$;
\item $\mathsf{Act}^{\sigma} = \{ \ast \}$;
\item $C^{\sigma} = \{ c \}$;
\item $(\mathcal R_s^{\sigma}, \ubar{\tau}_s^{\sigma} \leq c \leq \bar\tau_s^{\sigma}, \ast, \{c\}, \mathcal R_t^{\sigma}) \in E^{\sigma}$ if
$\mathcal X_{[\ubar{\tau}_s^{\sigma},\bar\tau_s^{\sigma}]}(\mathcal R_s^{\sigma}) \cap \mathcal R_t^{\sigma} \neq \emptyset$;
\item $\mathsf{Inv}^{\sigma}(\mathcal R_s^{\sigma}) = \{ c \mid 0 \leq c \leq \bar\tau_s^{\sigma} \}$ for all $s \in \{1,\dots,q\}$.
\end{itemize}
\end{definition}

Each location $\mathcal R_s^{\sigma}$ of this TA, is associated with a set of possible states $x$ of the system \eqref{eqn:lti-system}. 
We abuse slightly notation denoting both the location and the associated region with the same symbol $\mathcal R_s^\sigma$.
The suggestion in \cite{ArmanTCoN} is to partition the state space of the control system in conic regions pointed at the origin, each of which would be associated to a location of the TA.  The TA has one clock variable $c$ that represents the time elapsed since the last update. 
According to \cite{Tabuada}, given a fixed sampled state, the inter-sample time is uniquely defined, i.e.\ it is deterministic.
In general, when the sampled state is different, the inter-sample time is also different.
The notation $\ubar{\tau}_s^{\sigma}$ and $\bar\tau_s^{\sigma}$ represents the lower and upper bounds of the inter-sample time
for sampled states in $\mathcal R_s^{\sigma}$.
In \cite{ArmanTCoN} it is shown formally that the TA abstracts the timing behavior of the event-triggered system, implying that:
$$
\forall s\in\lbrace1,\ldots,q\rbrace, \forall x \in \mathcal R_s^{\sigma}: \tau_\sigma(x)\in[\ubar{\tau}_s^{\sigma}, \bar\tau_s^{\sigma}].
$$
\begin{remark}
In principle it could happen that $\tau_\sigma(x)=\infty$ for some states $x$ of the system. In practice, one would always impose a maximum time between transmissions to maintain a minimum level of feedback. This practical solution is suggested in \cite{ArmanTCoN} to guarantee having always $\bar\tau_s^{\sigma}<\infty$, as otherwise the TA model would become useless for scheduling purposes.
\end{remark}
From Definition \ref{def:lti2ta-arman} it is also trivial to see that if $\xi(t) = x \in \mathcal R_s^{\sigma}$ then $\xi(t+\tau_\sigma(x)) \in \mathcal R_j^{\sigma}$, with $\mathcal R_j^{\sigma}$ being one of the end locations in the set of edges with starting location $\mathcal R_s^{\sigma}$.
The outgoing edges of $\mathcal R_s^{\sigma}$ are enabled if the time elapsed since the last update is between $\ubar\tau_s^{\sigma}$ and $\bar\tau_s^{\sigma}$.
Only one action denoted by $\ast$ is present in this model, and since taking any edge is interpreted as updating the input value,
all edges are labeled with action $\ast$ and reset the clock variable.
Note that the system may remain in location $\mathcal R_s^{\sigma}$ for at most $\bar\tau_s^{\sigma}$ time units, as a triggering event is guaranteed to happen before that instant. 
A graphical representation of a simple TA of the form of those from Definition \ref{def:lti2ta-arman} is shown in Fig.\ \ref{fig:ta-lti-arman}.

\begin{figure}[!t]
\centering
\includegraphics{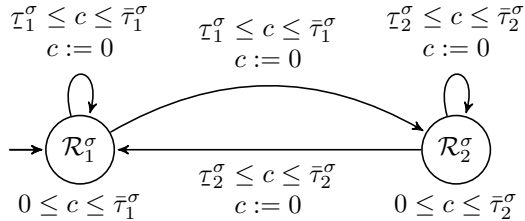}
\caption{A timed automaton modeling a control loop with event-triggered implementation. The unsourced arrow indicates the initial location. The internal action $\ast$ is omitted in the figure.}\label{fig:ta-lti-arman}
\end{figure}

\section{Scheduling of Event-Triggered Control Systems}
\label{sec:scheduling}

Consider a set of event-triggered networked control systems (NCSs) sharing a common communication channel (cf.\ Fig.\ \ref{fig:ncs}). Each control loop consists of a sensor, a plant, an actuator, and a controller, interconnected through the shared communication network.
Assume that the network can be used by at most one control loop at any time instant.
If several control loops request access to the channel while the network is in use a conflict arises, and at most
one control loop will be chosen nondeterministically to access the network.
While in time-triggered control systems these type of problems can be prevented by appropriate scheduling, when one or several control-loops are event-triggered a-priori scheduling is a much more challenging task because of the unknown update times.

In this section, we propose an approach based on NTGA to avoid such conflicts.
We consider schedulers that after each update of a control loop (transmission of measurements, computation of control and transmission of actuation signal to actuators) decide whether the next update time of each control loop should: 
\begin{itemize}
\item be based on a triggering mechanism selected from a set of finitely many triggering coefficients $\{\sigma_1,\ldots, \sigma_p\}$; or
\item forced to be at a pre-defined time (earlier than the minimum expected inter-sample time). 
\end{itemize}

We synthesize scheduling strategies by: constructing an NTGA from a set of NCSs (cf.\ Section \ref{sec:scheduling-model}),
characterizing the bad states, i.e.\ states corresponding to a communication conflict (cf.\ Section \ref{sec:scheduling-spec}), and 
finally synthesizing a supervising strategy ensuring that the NTGA avoids the bad states.

\subsection{Model: Network of Timed Game Automata}
\label{sec:scheduling-model}

In what follows, we describe the procedure employed to construct an NTGA from a given set of event-based NCSs.
We start constructing a TGA associated with the shared communication network.
Next, for each event-triggered control loop, we generate a TGA as a modification of the TA described in Section \ref{sec:modelprel-abs}.
Finally, the NTGA is obtained by taking the parallel composition of the TGA associated with the network and with all control loops.

\subsubsection{Communication Network}
\label{sec:scheduling-model-cu}

Denote the TGA corresponding to the shared communication network by $\mathsf{TGA}^{\mathit{net}}$ (cf.\ Definition \ref{def:tga-net}).
$\mathsf{TGA}^{\mathit{net}}$ has three locations $\mathit{Idle}$, $\mathit{InUse}$ and $\mathit{Bad}$, where the initial location is $\mathit{Idle}$ (cf.\ Fig.\ \ref{fig:tga-cu}).
The location $\mathit{Idle}$ represents the network being available,
$\mathit{InUse}$ represents the network being used by a control loop and
$\mathit{Bad}$ represents a conflict occured.
The active location changes from $\mathit{Idle}$ to $\mathit{InUse}$ when a control loop requests access to the channel to perform an update, 
which also forces the reset of the clock variable $c$.
The channel is occupied for $\Delta$ time units before the network is freed again to service the control tasks.
During this time, the active location is $\mathit{InUse}$, and after that time the active location changes to $\mathit{Idle}$.
When the active location is $\mathit{InUse}$ and
another control loop requests access then the active location changes to $\mathit{Bad}$. 
Once the network enters the location $\mathit{Bad}$, the network cannot leave the location, i.e.\ $\mathit{Bad}$ is an absorbing location.
Notice that this is a somewhat conservative model, as we consider every control loop occupies the channel the whole time $\Delta$. 
One could trivially adjust this simple model, and the subsequent work, to associate different occupancy times to different control loops.

\begin{definition}\label{def:tga-net}
Let $\Delta$ represent the maximum channel occupancy time,
a timed game automaton associated with the communication network is given by $\mathsf{TGA}^{\mathit{net}} = (L^{\mathit{net}},\ell_0^{\mathit{net}},\mathsf{Act}_c^{\mathit{net}},$ $\mathsf{Act}_u^{\mathit{net}},C^{\mathit{net}},E^{\mathit{net}},\mathsf{Inv}^{\mathit{net}})$ where
\begin{itemize}
\item $L^{\mathit{net}} = \{ \mathit{Idle}, \mathit{InUse}, \mathit{Bad} \}$;
\item $\ell_0^{\mathit{net}} = \mathit{Idle}$;
\item $\mathsf{Act}_c^{\mathit{net}} = \{ \ast \}$;
\item $\mathsf{Act}_u^{\mathit{net}} = \{ \mathit{up}? \}$;
\item $C^{\mathit{net}} = \{ c \}$;
\item $E^{\mathit{net}} = \{ ( \mathit{Idle}, \mathit{true}, \mathit{up}? , \{ c \} , \mathit{InUse} ),
( \mathit{InUse} , c = \Delta , \ast , \emptyset , \mathit{Idle} ), $\\
$( \mathit{InUse}, true, \mathit{up}? , \emptyset , \mathit{Bad} ),
( \mathit{Bad}, \mathit{true}, \mathit{up}? , \emptyset , \mathit{Bad} )
 \}$;
\item $\mathsf{Inv}^{\mathit{net}}(\mathit{InUse}) = \{c \mid 0 \leq c \leq \Delta \}$,\\ $\mathsf{Inv}^{\mathit{net}}(\mathit{Idle}) = \{c \mid c \geq 0 \}$, $\mathsf{Inv}^{\mathit{net}}(\mathit{Bad}) = \{c \mid c \geq 0 \}$.
\end{itemize}
The guard $\mathit{true}$ represents a condition that is always satisfied, for example $c \geq 0$.
\end{definition}

\begin{figure}[!t]
\centering
\includegraphics{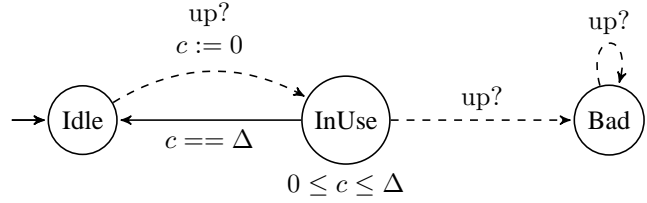}
\caption{A timed game automaton modeling the shared communication network. The solid and dashed arrows represent controllable and uncontrollable edges, respectively.}\label{fig:tga-cu}
\end{figure}

\subsubsection{Control Loops}
\label{sec:scheduling-model-lti}

Given a control loop,
we construct the timed game automata $\mathsf{TGA}^\mathit{cl}$ allowing a supervisor (scheduler) to either: force earlier controller updates than those dictated by the event-triggering mechanism, or choose a triggering coefficient for the event-triggering mechanism.

\begin{definition}\label{def:lti2ta-cl}
Consider a set of timed automata $\mathsf{TA}^{\sigma_j} = (L^{\sigma_j},\ell_0^{\sigma_j},\mathsf{Act}^{\sigma_j},C^{\sigma_j},E^{\sigma_j},\mathsf{Inv}^{\sigma_j})$ generated from an event-triggered control loop with triggering coefficient $\sigma_j\in]0,\bar\sigma[$ for $j \in \{1,\dots,p\}$ and assume that $\mathcal R_s^{\sigma_1} = \dots = \mathcal R_s^{\sigma_p}$ for all $s \in \{1,\dots,q\}$.\\
Consider also a set of earlier update time parameters $\{\ubar{d}_1,\bar d_1,\dots,\ubar{d}_q,\bar d_q\}$, such that $$\forall s \in \{1,\dots,q\}\,\exists\, j \in \{1,\dots,p\}:\,\bar d_s \leq \ubar{\tau}_s^{\sigma_j}.$$
Then, the timed game automata $\mathsf{TGA}^{\mathit{cl}}$ is given by\\
$\mathsf{TGA}^{\mathit{cl}} = (L^{\mathit{cl}},\ell_0^{\mathit{cl}},\mathsf{Act}_c^{\mathit{cl}},\mathsf{Act}_u^{\mathit{cl}},C^{\mathit{cl}},E^{\mathit{cl}},\mathsf{Inv}^{\mathit{cl}})$ where
\begin{itemize}
\item $L^{\mathit{cl}} = \bigcup_{j=1}^p L^{\sigma_j} \cup \bigcup_{s=1}^q \{ \mathcal R_s, \mathit{Ear}_s \}$;
\item $\ell_0^{\mathit{cl}} = \mathcal{R}_s$ such that $\xi(0) \in \mathcal{R}_s^{\sigma_1}$;
\item $\mathsf{Act}_c^{\mathit{cl}} =   \mathsf{Act}^{\sigma_1} \cup \bigcup_{j=1}^p \{a_j^{\mathit{cl}}\}$;
\item $\mathsf{Act}_u^{\mathit{cl}} = \{ \mathit{up}! \}$;
\item $C^{\mathit{cl}} = C^{\sigma_1}$;
\item $E^{\mathit{cl}} = \bigcup_{s=1}^q \bigcup_{t \in \mathcal E_s} \{ ( \mathit{Ear}_s, c = 0, \mathit{up}!, \emptyset, \mathcal R_t ) \}
\cup
\bigcup_{s=1}^q \bigcup_{j=1}^p \{ ( \mathcal R_s, c = 0, a_j^{\mathit{cl}}, \emptyset, \mathcal R_s^{\sigma_j} ),\\
(\mathcal R_s^{\sigma_j},\ubar d_s \leq c \leq \bar d_s, \ast,\{c\},\mathit{Ear}_s) \} \cup
\bigcup_{s=1}^q \bigcup_{j=1}^p \bigcup_{\{t | (\mathcal R_s \to \mathcal R_t) \in E^{\sigma_j}\}}$\\
$\{ ( \mathcal R_s^{\sigma_j}, \ubar{\tau}_s^{\sigma_j} \leq c \leq \bar\tau_s^{\sigma_j},\mathit{up}!,\{c\}, \mathcal R_t ) \}$;
\item $\mathsf{Inv}^{\mathit{cl}}(\mathcal R_s^{\sigma_j}) = \{ c \mid c \leq \bar \tau_s^{\sigma_j} \},\mathsf{Inv}^{\mathit{cl}}(\mathcal R_s) = \{ c \mid c = 0 \}$,\\
$\mathsf{Inv}^{\mathit{cl}}(\mathcal \mathit{Ear}_s) = \{ c \mid c = 0 \}.$
\end{itemize}
\end{definition}

In model just introduced, we use separate locations associated to each triggering coefficient and introduce the
additional locations  $\mathcal R_s$ and $\mathit{Ear}_s$ for $s \in \{1,\dots,q\}$. 
Both the locations $\mathcal R_s$ and $\mathcal R_s^{\sigma_j}$ represent that the sampled state is in $\mathcal R_s^{\sigma_1}$.
In location $\mathcal R_s$ the scheduler has not chosen the triggering coefficient, 
whereas in location $\mathcal R_s^{\sigma_j}$ the scheduler has chosen triggering coefficient $\sigma_j$.
Since the scheduler can choose the triggering coefficient,
the edges from $\mathcal{R}_s$ to $\mathcal{R}_s^{\sigma_j}$ are labeled with the controllable action $a_j$.
After choosing the triggering coefficient, the scheduler is allowed to either:
force earlier controller updates, or use the event-triggering mechanism (based on the chosen triggering coefficient).

If the scheduler decides to use the event-triggering mechanism,
while staying in location $\mathcal R_s^{\sigma_j}$,
the strategy is ``do nothing''.
This ensures that the outgoing edge to $\mathit{Ear}_s$ is not taken.
When the value of $c$ is between $\ubar\tau_s^{\sigma_j}$ and $\bar\tau_s^{\sigma_j}$, the event-triggering mechanism is activated.
In this case, the edges from $\mathcal{R}_s^{\sigma_j}$ to $\mathcal R_t$ labeled with the {uncontrollable} action $\mathit{up}!$ are enabled.
Recall that the scheduler cannot choose the exact update time when using the event-triggering mechanism.
This also implies that the scheduler cannot choose the region containing the next sampled state.

If the scheduler decides to force earlier controller updates,
the scheduler will take the edge to $\mathit{Ear}_s$ when that edge is enabled.
In this case, the scheduler is able to choose the exact update time.
Thus, the edges from $\mathcal{R}_s^{\sigma_j}$ to $\mathit{Ear}_s$ are labeled with the controllable action $\ast$.
In location $\mathit{Ear}_s$, the time cannot elapse and one of the outgoing edges has to be taken immediately.
Since the scheduler cannot choose the region containing the next sampled state, the outgoing edges of $\mathit{Ear}_s$ are labeled with the {uncontrollable} action $\mathit{up}!$. 
The outgoing edges are defined as follows:
there exists an edge from $\mathit{Ear}_s$ to $\mathcal{R}_t$ if $t \in \mathcal E_s := \{ t \mid \mathcal{X}_{[\ubar d_s,\bar d_s]}(\mathcal R_s^{\sigma_1}) \cap \mathcal R_t^{\sigma_1} \neq \emptyset \}$.
A graphical representation of a TGA generated by Definition \ref{def:lti2ta-cl} is shown in Fig.\ \ref{fig:tga-cl-ex}.

\begin{figure}[!t]
\centering
\includegraphics{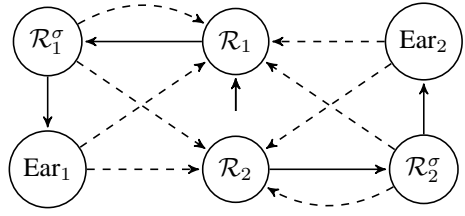}
\caption{A timed game automaton modeling a control loop with an event-triggered implementation. The labels over edges and location invariants are not shown for simplicity.}\label{fig:tga-cl-ex}
\end{figure}

In this subsection, we assume the initial conditions of the LTI system are a subset of a region.
If the initial conditions are intersected with a set of regions $\mathcal R_{\mathit{init}} \subseteq \{ \mathcal R_1^{\sigma_1}, \dots, \mathcal R_q^{\sigma_1} \}$,
we can modify the TGA generated by Definition \ref{def:lti2ta-cl} as follows.
Introduce a new location called $\mathcal R_0$ with invariant $\{ c = 0 \}$ and define $\mathcal R_0$ as the initial location.
Then, define edges from $\mathcal R_0$ to every location in $\mathcal R_{\mathit{init}}$ with guard $c = 0$, action $\circledast$ and without resetting the clock.
Finally, action $\circledast$ is defined as an uncontrollable action.
In the above modification, the environment has to choose one of the locations corresponding to initial conditions when the system is started.
A graphical representation of the TGA representing this situation is depicted in Fig.\ \ref{fig:tga-cl-r0-ex}.

\begin{figure}[!t]
\centering
\includegraphics{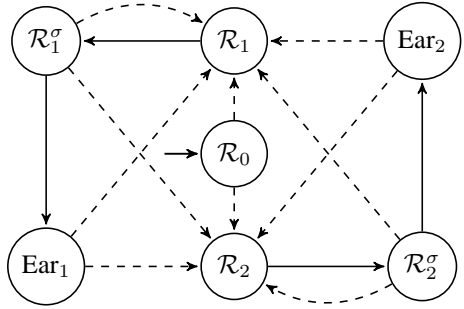}
\caption{A timed game automaton modeling a control loop where the initial states are intersected with both conic regions. The labels over edges and location invariants are not shown for simplicity.}\label{fig:tga-cl-r0-ex}
\end{figure}

\begin{proposition}
Switching between different triggering coefficients or triggering earlier does not hinder stability.
\end{proposition}
\begin{proof}
Consider Lyapunov function $V : \mathbb{R}^n \to \mathbb R_{\geq 0}$ satisfying $\dot{V}(\xi(t))\leq -\lambda_c V(\xi(t))$, with $\lambda_c>0$ for the system (\ref{eqn:lti-system}) with continuous feedback $u(t)=K\xi(t)$. 
It has been shown in \cite{Tabuada} that selecting a triggering coefficient $0<\sigma < \bar\sigma$, with $\bar\sigma$ an appropriate constant depending on the LTI dynamics and the state-feedback gain, the event-triggered controller implementation (\ref{eqn:lti-sample-hold})-(\ref{eqn:tabtrig}) satisfies $\forall t:\, \dot{V}(\xi(t))\leq -\lambda_e(\sigma) V(\xi(t))$, with 
$\lambda_c>\lambda_e(\sigma_1)>\lambda_e(\sigma_2)>0$ for $0<\sigma_1<\sigma_2<\bar\sigma$. 

In fact, the triggering mechanism guarantees that $\dot{V}(\xi(t))\leq -\lambda_e(\sigma) V(\xi(t))$ in the interval $t\in[t_k, t_{k}+\tau_\sigma(x)]$,  for all $x=\xi(t_k)$.
Since $\sigma_j < \bar\sigma$ for all $j \in \{1,\dots,p\}$, switching between different triggering coefficients guarantees that $\forall t:\, \dot{V}(\xi(t))\leq -\lambda_e{\tilde\sigma}V(\xi(t))$, with $\tilde\sigma:=\max_{j \in \{1,\dots,p\}} \sigma_j$.

Finally, if the system is forced to employ earlier triggering the assumption:
$\forall s \in \{1,\dots,q\}\,\exists\, j \in \{1,\dots,p\}:\,\bar d_s \leq \ubar{\tau}_s^{\sigma_j}$
guarantees that the update occurs in the interval $[t_k, t_k+\tau_{\tilde\sigma}(x)]$, and thus $\forall t:\, \dot{V}(\xi(t))\leq -\lambda_e{\tilde\sigma}V(\xi(t))$, which concludes the proof.
\end{proof}

As it was mentioned in the beginning of Section \ref{sec:scheduling-model},
the NTGA associated with a set of NCSs, denoted by $\mathsf{TGA}^{\mathit{NCSs}}$, is obtained by taking the parallel composition of the TGA associated with the network and with all control loops.
In other words,
$\mathsf{TGA}^{\mathit{NCSs}} := \mathsf{TGA}^{\mathit{net}} \mid \mathsf{TGA}^{\mathit{cl1}} \mid \dots \mid \mathsf{TGA}^{\mathit{clN}}$ where
$\mathsf{TGA}^{\mathit{cli}}$, $i \in \{1,\dots,N\}$, represents the TGA associated with the $i$-th control loop.
The state of $\mathsf{TGA}^{\mathit{NCSs}}$ is described by a $(2N+2)$-tuple $(\ell_{\mathit{net}},\ell_1,\dots,\ell_N,u_{\mathit{net}},u_1,\dots,u_N)$ where
$\ell_{\mathit{net}}$ is the location of $\mathsf{TGA}^{\mathit{net}}$,
$\ell_i$ is the location of $\mathsf{TGA}^{\mathit{cli}}$,
$u_{\mathit{net}}$ is the clock assignment of $\mathsf{TGA}^{\mathit{net}}$ and
$u_i$ is the clock assignment of $\mathsf{TGA}^{\mathit{cli}}$,
for $i \in \{1,\dots,N\}$.

\subsection{Specification: Safety}
\label{sec:scheduling-spec}

We are interested in finding a strategy such that the trajectories of the NTGA never enter the states corresponding to a conflict.
Recall that a conflict corresponds to the following situation:
a control loop is requesting updates when the communication network is busy.
In our NTGA model, conflicts are captured by
the active location of $\mathsf{TGA}^{\mathit{net}}$ becoming $\mathit{Bad}$.
Thus the set of states we aim at avoiding $\mathcal A$ contains all states such that the location of $\mathsf{TGA}^{\mathit{net}}$ is $\mathit{Bad}$, i.e.
\begin{displaymath}
\mathcal A = \{ (\ell_{\mathit{net}},\ell_1,\dots,\ell_N,u_{\mathit{net}},u_1,\dots,u_N) \mid \ell_{\mathit{net}} = \mathit{Bad} \}.
\end{displaymath}


\subsection{Limiting the Consecutive Earlier Updates}
\label{sec:scheduling-limit}

If we allow the scheduler to force earlier controller updates, nothing prevents the scheduler from always using such type of updates and never employ the event-triggering mechanism.
In this section, we discuss an approach to prevent the undesired schedule by limiting the consecutive earlier updates.
In this approach, once a pre-specified limit has been reached, the scheduler is forced to use an event-triggering mechanism.

For this purpose we employ global integer variables, which are an extended feature of UPPAAL-Tiga modeling language to ease the modeling task but not part of the standard definition of TGA (cf.\ Definition \ref{def:tga}).
We define a global integer constant $\mathtt{earMax}$ representing the maximum consecutive earlier updates, and a global integer variable $\mathtt{earNum}$ to be used as a counter of consecutive earlier updates. A variable is global if it can be accessed by all TGAs.
Finally, the resulting TGA is defined as follows.

\begin{definition}\label{def:lti2ta-cl-max}
Consider a set of timed automata $\mathsf{TA}^{\sigma_j} = (L^{\sigma_j},\ell_0^{\sigma_j},\mathsf{Act}^{\sigma_j},C^{\sigma_j},E^{\sigma_j},\mathsf{Inv}^{\sigma_j})$ generated from an event-triggered control loop with triggering coefficient $\sigma_j\in]0,\bar\sigma[$ for $j \in \{1,\dots,p\}$ and assume that $\mathcal R_s^{\sigma_1} = \dots = \mathcal R_s^{\sigma_p}$ for all $s \in \{1,\dots,q\}$.\\
Consider also some constant $\mathtt{earMax}$ and a set of earlier update time parameters $\{\ubar{d}_1,\bar d_1,\dots,\ubar{d}_q,\bar d_q\}$, such that $$\forall s \in \{1,\dots,q\}\,\exists\, j \in \{1,\dots,p\}:\,\bar d_s \leq \ubar{\tau}_s^{\sigma_j}.$$
Then, the timed game automata with options for earlier update, choice of triggering coefficients and limiting the consecutive earlier updates is given by
$\mathsf{TGA}^{\mathit{clim}} = (L^{\mathit{clim}},\ell_0^{\mathit{clim}},\mathsf{Act}_c^{\mathit{clim}},\mathsf{Act}_u^{\mathit{clim}},C^{\mathit{clim}},E^{\mathit{clim}},\mathsf{Inv}^{\mathit{clim}})$ where
\begin{itemize}
\item $L^{\mathit{clim}} = \bigcup_{j=1}^p L^{\sigma_j} \cup \bigcup_{s=1}^q \{ \mathcal R_s, \mathit{Ear}_s \}$;
\item $\ell_0^{\mathit{clim}} = \mathcal{R}_s$ such that $\xi(0) \in \mathcal{R}_s^{\sigma_1}$;
\item $\mathsf{Act}_c^{\mathit{clim}} = \{ \mathit{up}! \} \cup \mathsf{Act}^{\sigma_1} \cup \bigcup_{j=1}^p \{a_j^{\mathit{clim}}\}$;
\item $\mathsf{Act}_u^{\mathit{clim}} = \emptyset$;
\item $C^{\mathit{clim}} = C^{\sigma_1}$;
\item $E^{\mathit{clim}} = \bigcup_{s=1}^q \bigcup_{t \in \mathcal E_s} \{ ( \mathit{Ear}_s, c = 0, \mathit{up}!, \emptyset, \mathcal R_t ) \}
\cup$\\
$\bigcup_{s=1}^q \bigcup_{j=1}^p \{ ( \mathcal R_s, c = 0, a_j^{\mathit{clim}}, \emptyset, \mathcal R_s^{\sigma_j} ),$\\
$(\mathcal R_s^{\sigma_j},(\ubar d_s \leq c \leq \bar d_s)
\wedge (\mathtt{earNum} < \mathtt{earMax}),\ast,$\\
$\{c\} \wedge (\mathtt{earNum} := \mathtt{earNum} + 1),\mathit{Ear}_s) \} \cup \bigcup_{s=1}^q \bigcup_{j=1}^p \bigcup_{\{t \mid (\mathcal R_s \to \mathcal R_t) \in E^{\sigma_j}\}}$\\
$\{ ( \mathcal R_s^{\sigma_j}, \ubar{\tau}_s^{\sigma_j} \leq c \leq \bar\tau_s^{\sigma_j},\mathit{up}!,\{c\} \wedge (\mathtt{earNum} := 0), \mathcal R_t ) \}$;
\item $\mathsf{Inv}^{\mathit{clim}}(\mathcal R_s^{\sigma_j}) = \{ c \mid c \leq \bar \tau_s^{\sigma_j} \},
\mathsf{Inv}^{\mathit{clim}}(\mathcal R_s) = \{ c \mid c = 0 \}$.
\end{itemize}
\end{definition}

There are two differences between Definition \ref{def:lti2ta-cl-max} and Definition \ref{def:lti2ta-cl}.
First, in the edges from $\mathcal R_s^{\sigma_j}$ to $\mathit{Ear}_s$, we add condition $\mathtt{earNum} < \mathtt{earMax}$ to the guard and add statement $\mathtt{earNum} := \mathtt{earNum} + 1$ to the reset.
The additional condition is used to guarantee that the counter of consecutive earlier updates is always smaller than or equal to its maximum.
The statement on the reset is used to increase the counter variable by one, once an earlier update happens.
Recall that an earlier update happens when one of these edges is taken (cf.\ Section \ref{sec:scheduling-model-lti}).
Second, in the edges from $\mathcal R_s^{\sigma_j}$ to $\mathcal R_t$,
we add the statement $\mathtt{earNum} := 0$ to the reset.
Recall that taking these edges represents the event-triggering mechanism is used (cf.\ Section \ref{sec:scheduling-model-lti}).
Thus the counter of consecutive earlier updates is reset to zero.
Notice that the variable $\mathtt{earNum}$ takes values in $\{0,1,\dots,\mathtt{earMax}\}$.

\begin{remark}
Note that with the presented implementation, either of the control loops may exhibit an arbitrary number of consecutive earlier triggerings.  This is because the maximum number of consecutive earlier triggerings being a global counter. The counter is reset to zero whenever any of the control loops runs in event-triggered fashion. By employing more counters, one could easily generalize this idea to limit the number of consecutive earlier triggerings for each loop.
\end{remark}
After these modifications, the NTGA associated to the set of NCSs becomes
$\mathsf{TGA}^{\mathit{NCSs}} := \mathsf{TGA}^{\mathit{net}} \mid \mathsf{TGA}^{\mathit{clim1}} \mid \dots \mid \mathsf{TGA}^{\mathit{climN}}$ where
$\mathsf{TGA}^{\mathit{climi}}$ represents the TGA associated with the $i$-th control loop for $i \in \{1,\dots,N\}$.
In this new NTGA, the state of $\mathsf{TGA}^{\mathit{NCSs}}$ is described by a $(2N+3)$-tuple
$(\ell_{\mathit{net}},\ell_1,\dots,\ell_N,u_{\mathit{net}},u_1,\dots,u_N,\mathtt{earNum})$ which includes the additional 
counter $\mathtt{earNum}$. It follows that the bad states $\mathcal A$ are now given by
\begin{displaymath}
\{ (\ell_{\mathit{net}},\ell_1,\dots,\ell_N,u_{\mathit{net}},u_1,\dots,u_N,\mathtt{earNum}) \mid \ell_{\mathit{net}} = \mathit{Bad} \}.
\end{displaymath}

\subsection{Scheduler Operation}
\label{sec:scheduling-description}


Formally, a scheduler implements a strategy $f$, see Definition \ref{def:strategy}, for the NTGA $\mathsf{TGA}^{\mathit{NCSs}}$.
The strategy $f$ is applied to $\mathsf{TGA}^{\mathit{NCSs}}$ providing, based on the run $\rho$ of $\mathsf{TGA}^{\mathit{NCSs}}$ up to that time instant
the controllable action $f(\rho)$ that guarantees the satisfaction of the desired specification.

This means in practice that after each discrete transition of the NTGA, i.e. every time a transmission is placed on the network, first the strategy
chooses a triggering coefficient. 
Then the strategy decides which control loop is updated and also its update mechanism: early or event triggered.
After such a transition, and possibly after some time elapses, the environment chooses the conic region containing the next sampled state, which results in a discrete transition of the NTGA, and the procedure is repeated.

\begin{example}
Let us illustrate the use of strategies on an example consisting of two control loops, two triggering coefficients $\{ \sigma_1, \sigma_2 \}$ and an option for earlier updates. The initial location of the first and second control loop is $\mathcal{R}_1$ and $\mathcal{R}_2$, respectively. Initially the run of $\mathsf{TGA}^{\mathit{NCSs}}$ is 
$$\rho_0 = (\mathit{Idle},\mathcal{R}_1,\mathcal{R}_2,0,0,0).$$ 
After each update, the scheduler selects a triggering coefficient, according to the strategy $f$. 
Suppose that the scheduler chooses $\sigma_2$ for the first control loop, i.e.\ $f(\rho_0) = a_2^{\mathit{cl1}}$. The resulting run is
$$\rho_1 = \rho_0 \rTo^{a_2^{\mathit{cl1}}}_{\mathit{TS}} (\mathit{Idle},\mathcal{R}_1^{\sigma_2},\mathcal{R}_2,0,0,0).$$
If the scheduler chooses $\sigma_1$ for the second control loop, i.e.\ $f(\rho_1) = a_1^{\mathit{cl2}}$, the run becomes
$$
\begin{array}{rl}
& \rho_2 = \rho_1 \rTo^{0}_{\mathit{TS}} (\mathit{Idle},\mathcal{R}_1^{\sigma_2},\mathcal{R}_2,0,0,0) \rTo^{a_1^{\mathit{cl2}}}_{\mathit{TS}} 
 (\mathit{Idle},\mathcal{R}_1^{\sigma_2},\mathcal{R}_2^{\sigma_1},0,0,0).
\end{array}
$$ 
Then the scheduler follows the strategy to decide which control loop is updated and also its update mechanism: early or event triggered.
Suppose that the strategy decides to update the first control loop earlier at time $\ubar d_1$.
First, the scheduler delays the system
$$\rho_3 = \rho_2 \rTo^{\ubar d_1}_{\mathit{TS}} (\mathit{Idle},\mathcal{R}_1^{\sigma_2},\mathcal{R}_2^{\sigma_1},\ubar d_1,\ubar d_1,\ubar d_1).$$ 
Then an earlier update is performed, i.e.\ $f(\rho_3) = \ast$.
Notice that action $\ast$ is the internal action associated with the first control loop.
We have run 
$$\rho_4 = \rho_3 \rTo^{\ast}_{\mathit{TS}} (\mathit{Idle},\mathit{Ear}_1,\mathcal{R}_2^{\sigma_1},\ubar d_1,0,\ubar d_1).$$ 
Since time cannot elapse in $\mathit{Ear}_1$, the environment has to choose the conic region containing the next sampled state immediately.
If the environment chooses $\mathcal R_3$, the result is run 
$$
\begin{array}{rl}
& \rho_5 = \rho_4 \rTo^{0}_{\mathit{TS}} (\mathit{Idle},\mathit{Ear}_1,\mathcal{R}_2^{\sigma_1},\ubar d_1,0,\ubar d_1) \rTo^{\mathit{up}}_{\mathit{TS}} 
 (\mathit{InUse},\mathcal R_3,\mathcal{R}_2^{\sigma_1},0,0,\ubar d_1).
\end{array}
$$
Notice that $\mathsf{TGA}^{\mathit{net}}$ and the TGA associated with the first control loop move simultaneously via synchronizing action $\mathit{up}$.
Input action $\mathit{up}?$ belongs to $\mathsf{TGA}^{\mathit{net}}$, whereas output action $\mathit{up}!$ belongs to the TGA corresponding to the first control loop.
Then the scheduler follows the strategy to select a triggering coefficient for the first control loop, for example $\sigma_1$, i.e.\ $f(\rho_5) = a_1^{\mathit{cl1}}$.
The network is available again after $\Delta$ time units, resulting in the runs:
$$
\begin{array}{rl}
& \rho_6 = \rho_5 \rTo^{0}_{\mathit{TS}} (\mathit{InUse},\mathcal R_3,\mathcal{R}_2^{\sigma_1},0,0,\ubar d_1) \rTo^{a_1^{\mathit{cl1}}}_{\mathit{TS}} \\
& \qquad (\mathit{InUse},\mathcal R_3^{\sigma_1},\mathcal{R}_2^{\sigma_1},0,0,\ubar d_1) \rTo^{\Delta}_{\mathit{TS}} 
 (\mathit{InUse},\mathcal R_3^{\sigma_1},\mathcal{R}_2^{\sigma_1},\Delta,\Delta,\ubar d_1+ \Delta),
\end{array}
$$
while the network is being used, and 
$$\rho_7 = \rho_6 \rTo^{\ast}_{\mathit{TS}} (\mathit{Idle},\mathcal R_3^{\sigma_1},\mathcal{R}_2^{\sigma_1},\Delta,\Delta,\ubar d_1+ \Delta)$$
once the network is released.
Notice that action $\ast$ is the internal action associated with $\mathsf{TGA}^{\mathit{net}}$.
\end{example}

Note that this kind of scheduler is a centralized scheduler that needs to have a perfect overview of the transmissions placed on the network, and the control loop responsible for it. Furthermore, given that the locations of $\mathsf{TGA}^{\mathit{NCSs}}$ are related to the actual sampled states transmitted through the network, the scheduler also needs to be able to read the content of the transmitted data.



\section{Case Study}
\label{sec:case}

We showcase the results in an example
comprising two event-triggered NCSs sharing the same communication network.
The first control loop is given by \cite[p.\ 1683]{Tabuada}
\begin{equation}
\begin{array}{rl}
& \dot\xi =
\begin{bmatrix}
0 & 1\\
-2 & 3
\end{bmatrix}
\xi +
\begin{bmatrix}
0 \\ 1
\end{bmatrix}
\upsilon,\\[2ex]
& \upsilon =
\begin{bmatrix}
1 & -4
\end{bmatrix}
\xi.
\end{array}
\label{eqn:lti-ex-1}
\end{equation}
The second control loop is given by \cite[p.\ 1699]{Hetel2011}
\begin{equation}
\begin{array}{rl}
& \dot\xi =
\begin{bmatrix}
-0.5 & 0\\
0 & 3.5
\end{bmatrix}
\xi +
\begin{bmatrix}
1 \\ 1
\end{bmatrix}
\upsilon,\\[2ex]
& \upsilon =
\begin{bmatrix}
1.02 & -5.62
\end{bmatrix}
\xi.
\end{array}
\label{eqn:lti-ex-2}
\end{equation}

In the sequel, we discuss two experimental results for the above example.
Each experiment is characterized by four parameters:
the number of conic regions $q$,
the set of triggering coefficients $\{ \sigma_1, \dots, \sigma_q \}$,
the set of earlier update parameters $\{ \ubar d_1, \bar d_1, \dots, \ubar d_q, \bar d_q \}$ and
maximum consecutive earlier triggering $\mathtt{earMax}$.


%
%
%
%
\begin{figure}[!t]
\centering
\includegraphics{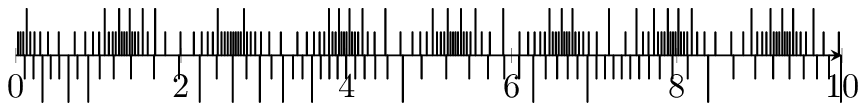}\\[1ex]
\includegraphics{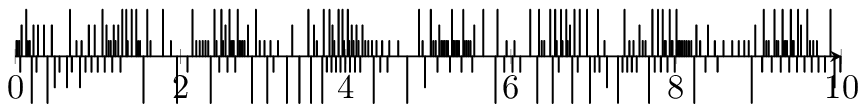}
\caption{Status of the shared communication network up to 10 time units.
The bars on the top and on the bottom of the $x$ axis represent the network is being used by \eqref{eqn:lti-ex-2} and \eqref{eqn:lti-ex-1}, respectively.
The top and bottom plots represent the result of experiments discussed in Sections \ref{sec:case-1} and \ref{sec:case-2}, respectively.}\label{fig:sim-Tab-Het-etc-trig}
\end{figure}
%
%

%
%

\subsection{Limiting the Consecutive Earlier Updates}
\label{sec:case-1}

In this experiment,
the minimum channel occupancy time is $\Delta = 0.005$,
the number of conic regions is $q = 200$ and
there is one triggering coefficient $\sigma_1 = 0.05$.
The input value can be updated 0.005 time units before the lower bound in all regions, i.e.\ $\ubar d_s = \ubar\tau_s - 0.005$ and $\bar d_s = \ubar\tau_s$ for all $s \in \{1,\dots,q\}$.
The maximum consecutive earlier triggering is $\mathtt{earMax} = 4$.

We create a model in UPPAAL-Tiga according to Definition \ref{def:lti2ta-cl-max}: the TGA for \eqref{eqn:lti-ex-1}, \eqref{eqn:lti-ex-2} and the shared communication network are denoted $\mathtt{tgaT}$, $\mathtt{tgaH}$ and $\mathtt{net}$, respectively.
The specification is given by
\begin{verbatim}
control: A[] not( net.Bad )
\end{verbatim}
A strategy is generated with UPPAAL-Tiga to satisfy this specification.
To illustrate the type of strategies synthesized, we show in the following a fragment of the strategy generated by UPPAAL-Tiga for the situation
in which the locations of $\mathtt{tgaT}$, $\mathtt{tgaH}$ and $\mathtt{net}$ are $\mathcal R_1^{\sigma_1}$, $\mathcal R_1^{\sigma_1}$ and $\mathit{Idle}$, respectively.
\begin{verbatim}
State: ( tgaT.R1a1 tgaH.R1a1 net.Idle )
  earNum=3
When you are in (25<=tgaH.c && tgaT.c<65 && tgaH.c-tgaT.c<=-35)
  || (85<tgaT.c && 25<=tgaH.c && tgaT.c<105 && tgaH.c<=30)
  || (38<tgaT.c && 25<=tgaH.c && tgaT.c-tgaH.c<=30 && tgaH.c<=30)
  || (25<=tgaH.c && tgaT.c<31 && tgaH.c-tgaT.c<=-5)
  || (25<=tgaH.c && tgaT.c-tgaH.c<=-5 && tgaH.c<=30),
take transition tgaH.R1a1->tgaH.Ear1
  { c >= 25 && c <= 30 && earNum < earMax, up!, 1 }
  net.Idle->net.InUse { 1, up?, c := 0 }
When you are in (105<=tgaT.c && tgaT.c <=111 && tgaH.c<25),
take transition tgaT.R1a1->tgaT.Ear1
  { c >= 105 && c <= 111 && earNum < earMax, up!, 1 }
  net.Idle->net.InUse { 1, up?, c := 0 }
\end{verbatim}
As shown above, two different conditions, based on the clock values of $\mathtt{tgaT}$, clock values of $\mathtt{tgaH}$ and the difference of clock values in $\mathtt{tgaT}$ and $\mathtt{tgaH}$, can be appreciated:
if the first one is satisfied, an early update is forced for $\mathtt{tgaH}$ where the inter-sample time is between 25 and 30;
if the second condition is satisfied, an early update is forced for $\mathtt{tgaT}$ where the inter-sample time is between 105 and 111.
If none of the conditions are satisfied, no early update is forced, i.e.\ the strategy is to let time elapses for both loops.

The strategy generated by UPPAAL-Tiga was applied to the two NCSs \eqref{eqn:lti-ex-1}-\eqref{eqn:lti-ex-2}, with both systems initialized at the state $[1,100]^T$, 
corresponding to $\mathcal R_1$ in both of the timing abstractions for the systems.
The network status is shown in Fig.\ \ref{fig:sim-Tab-Het-etc-trig} (top), where
long and short bars represent event-triggered and earlier update mechanisms, respectively.
Note that while either of the control loops may exhibit an arbitrary number of consecutive triggerings, the maximum number of consecutive earlier triggerings is respected to be below 4 as this counter is a shared (global) one that is reset to zero whenever any of the two loops run in event-triggered fashion.
During the time horizon of 10, the input of \eqref{eqn:lti-ex-1} is updated 63 times consisting of 50 earlier updates and 13 event-triggering mechanisms.
For \eqref{eqn:lti-ex-2}, the input is updated 152 times consisting of 121 earlier updates and 31 event-triggering mechanisms.
The state and input trajectories are shown in Fig.\ \ref{fig:sim-Tab-Het-etc-early-v2a}.

\begin{figure}[!t]
\centering
\includegraphics{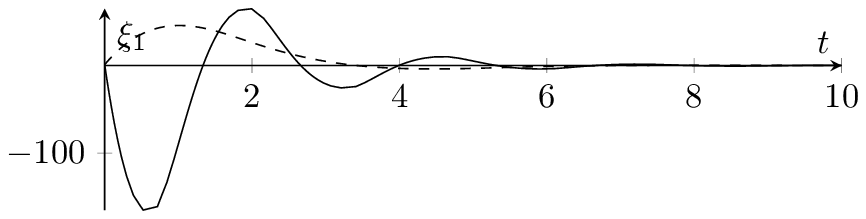}\\
\includegraphics{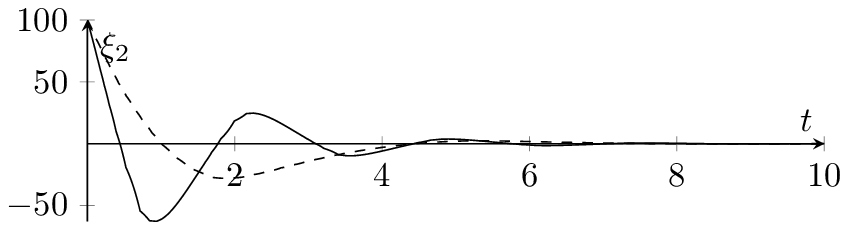}\\[1ex]
\includegraphics{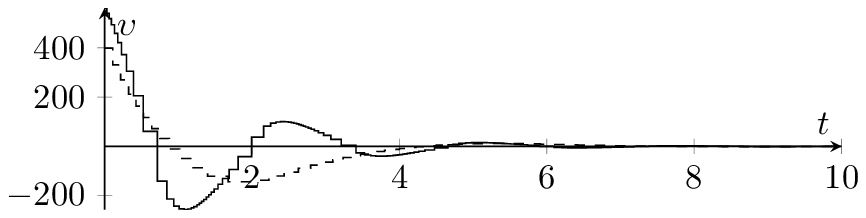}
\caption{The state and input trajectories for the experiment in Section \ref{sec:case-1}. The solid and dashed lines are the trajectories of \eqref{eqn:lti-ex-2} and \eqref{eqn:lti-ex-1}, respectively.}\label{fig:sim-Tab-Het-etc-early-v2a}
\end{figure}

\subsection{Choice of Triggering Coefficients}
\label{sec:case-2}

In this experiment,
the minimum channel occupancy time is $\Delta = 0.005$,
the number of conic regions is $q = 200$ and
there are three triggering coefficients $\sigma_1 = 0.01$, $\sigma_2 = 0.03$ and $\sigma_3 = 0.09$.

We create again a model in UPPAAL-Tiga according to
Definition \ref{def:lti2ta-cl}: the TGA for \eqref{eqn:lti-ex-1}, \eqref{eqn:lti-ex-2} and the shared communication network are denoted $\mathtt{tgaT}$, $\mathtt{tgaH}$ and $\mathtt{net}$, respectively.
The specification is the same as in Section \ref{sec:case-1} and again we generate a strategy using UPPAAL-Tiga.
The following is a fragment of the strategy generated by UPPAAL-Tiga when the location of $\mathtt{tgaT}$, $\mathtt{tgaH}$ and $\mathtt{net}$ is $\mathcal R_{37}$, $\mathcal R_{38}^{\sigma_1}$ and $\mathit{Idle}$, respectively.
\begin{verbatim}
State: ( tgaT.R37 tgaH.R38a1 net.Idle )
When you are in (tgaT.c==0 && 145<tgaH.c && tgaH.c<=154),
take transition tgaT.R37->tgaT.R37a1 { c == 0, tau, 1 }
When you are in (tgaT.c==0 && 65<=tgaH.c && tgaH.c<=102),
take transition tgaT.R37->tgaT.R37a2 { c == 0, tau, 1 }
When you are in (tgaT.c==0 && 102<tgaH.c && tgaH.c<=145) 
  || (tgaT.c==0 && 5<=tgaH.c && tgaH.c<65),
take transition tgaT.R37->tgaT.R37a3 { c == 0, tau, 1 }
\end{verbatim}
Notice that now there are three conditions:
if the $i$-th condition is satisfied, the location of $\mathtt{tgaT}$ is forced to transit to $\mathcal R_{37}^{\sigma_i}$ for $i \in \{1,2,3\}$.

The strategy generated by UPPAAL-Tiga is applied to the NCSs \eqref{eqn:lti-ex-1}-\eqref{eqn:lti-ex-2}, both with the state initialized at $[1,100]^T$, corresponding in both cases with the initial location is $\mathcal R_1$.
The network status is shown in Fig.\ \ref{fig:sim-Tab-Het-etc-trig} (bottom).
Short, medium and long bars represent event-triggered with triggering coefficient equals $\sigma_1$, $\sigma_2$ and $\sigma_3$, respectively.
During the first 10 time units, the input of \eqref{eqn:lti-ex-1} is updated 84 times consisting of 56 updates using $\sigma_1$, 6 updates using $\sigma_2$ and 22 updates using $\sigma_3$.
For \eqref{eqn:lti-ex-2}, the input is updated 182 times consisting of 106 updates using $\sigma_1$, 26 updates using $\sigma_2$ and 50 updates using $\sigma_3$.
The state and input trajectories are shown in Fig.\ \ref{fig:sim-Tab-Het-etc-const}.

\begin{figure}[!t]
\centering
\includegraphics{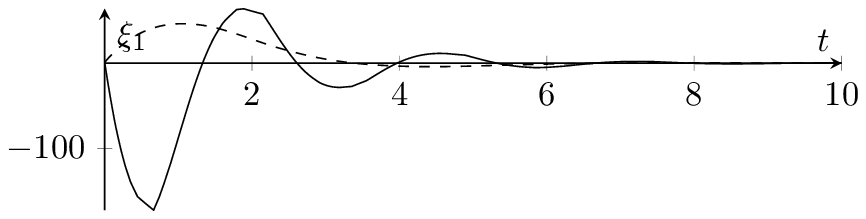}\\
\includegraphics{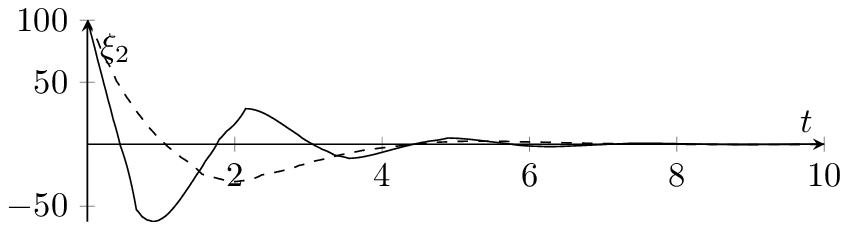}\\[1ex]
\includegraphics{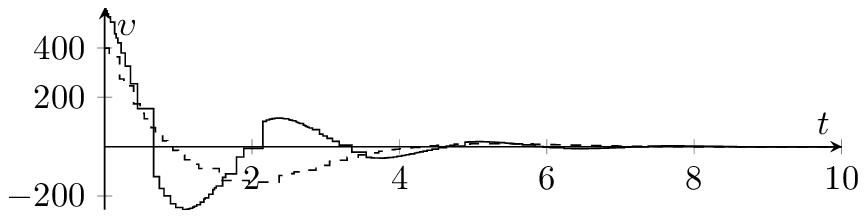}
\caption{The state and input trajectories for the experiment in Section \ref{sec:case-2}. The solid and dashed lines are the trajectories of \eqref{eqn:lti-ex-2} and \eqref{eqn:lti-ex-1}, respectively.}\label{fig:sim-Tab-Het-etc-const}
\end{figure}

%

\section{Discussion and Future Work}
\label{sec:concl}

We have provided an approach to synthesize conflict-free scheduling policies for sets of networked control systems (NCSs) with the possibilities of updating the input value according to an event-triggering mechanism, selectable from a set of them, or earlier than the time dictated by such an event-triggering rule.
As indicated in Section \ref{sec:scheduling-limit} the main limitations of this proposed scheduling scheme is its centralized nature, and the fact that the scheduler needs to be able to read the content of the messages sent through the network. The approach is nonetheless applicable to many setups encountered in practice in which different control systems are interconnected through a bus, e.g.\ CAN, EtherCAT or FlexRay. In such systems, every element connected to the bus can see the traffic flowing through the network. While it may be the case that the precise content of messages is not available (e.g.\ for potential security reasons), it is worth noting that for the scheduler only the abstracted state, i.e. the region $\mathcal{R}_s$, is relevant. Therefore we can envisage implementations or practical applications in which this sort of scheduling could be easily adopted.

In wireless settings, the event-triggered paradigm offers great benefits for energy consumption reduction, but network topologies can be in general more complex than a simple bus type of configuration. Therefore, interesting extensions of this work to allow decentralized scheduling, possibly including network topological constraints, would enable broader applicability of these techniques.  Current and future work is focusing on these issues, extensions of the abstraction of the timing of event-triggered systems beyond LTI systems with state-feedback, and on the implementation of a tool-box automating the whole timing abstraction and scheduler synthesis proposed in \cite{ArmanTCoN} and the current paper respectively.

\bibliographystyle{plain}

\begin{thebibliography}{10}
	\providecommand{\url}[1]{#1}
	\csname url@samestyle\endcsname
	\providecommand{\newblock}{\relax}
	\providecommand{\bibinfo}[2]{#2}
	\providecommand{\BIBentrySTDinterwordspacing}{\spaceskip=0pt\relax}
	\providecommand{\BIBentryALTinterwordstretchfactor}{4}
	\providecommand{\BIBentryALTinterwordspacing}{\spaceskip=\fontdimen2\font plus
		\BIBentryALTinterwordstretchfactor\fontdimen3\font minus
		\fontdimen4\font\relax}
	\providecommand{\BIBforeignlanguage}[2]{{%
			\expandafter\ifx\csname l@#1\endcsname\relax
			\typeout{** WARNING: IEEEtran.bst: No hyphenation pattern has been}%
			\typeout{** loaded for the language `#1'. Using the pattern for}%
			\typeout{** the default language instead.}%
			\else
			\language=\csname l@#1\endcsname
			\fi
			#2}}
	\providecommand{\BIBdecl}{\relax}
	\BIBdecl
	
	\bibitem{Liu2009}
	J.~Liu, A.~Gusrialdi, D.~Obradovic, and S.~Hirche, ``Study on the effect of
	time delay on the performance of distributed power grids with networked
	cooperative control,'' in \emph{Proceedings of the 1st {IFAC} Workshop on Estimation and
		Control of Networked Systems}, 2009, pp. 168--173.
	
	\bibitem{Thompson2004}
	H.~A. Thompson, ``Wireless and internet communications technologies for
	monitoring and control,'' \emph{Control Engineering Practice}, vol.~12,
	no.~6, pp. 781--791, 2004.
	
	\bibitem{Lehmann2011}
	D.~Lehmann and J.~Lunze, ``Extension and experimental evaluation of an
	event-based state-feedback approach,'' \emph{Control Engineering Practice},
	vol.~19, no.~2, pp. 101--112, 2011.
	
	\bibitem{special_proceedings}
	P.~Antsaklis and J.~Baillieul, ``Special issue on technology of networked
	control systems,'' \emph{Proceedings of the IEEE}, vol.~95, no.~1, pp. 5--8, Jan. 2007.
	
	\bibitem{reviewCUDR}
	G.~Nair, F.~Fagnani, S.~Zampieri, and R.~Evans, ``Feedback control under data
	rate constraints: An overview,'' \emph{Proceedings of the IEEE}, vol.~95, no.~1, pp.
	108--137, Jan. 2007.
	
	\bibitem{Hespanha07asurvey}
	J.~P. Hespanha, P.~Naghshtabrizi, and Y.~Xu, ``A survey of recent results in
	networked control systems,'' \emph{Proceedings of the IEEE}, vol.~95, pp. 138--162, Jan.
	2007.
	
	\bibitem{HandbookNECS}
	D.~Hristu-Varsakelis and W.~S. Levine, Eds., \emph{Handbook of networked and
		embedded control systems}, ser. Control Engineering.\hskip 1em plus 0.5em
	minus 0.4em\relax Boston, MA: Birkh\"auser Boston Inc., 2005.
	
	\bibitem{RUNES}
	K.-E. \AA{}rz\'{e}n, A.~Bicchi, S.~Hailes, K.~Johansson, and J.~Lygeros, ``On
	the design and control of wireless networked embedded systems,'' in
	\emph{IEEE Computer Aided Control System Design,}, Oct. 2006, pp. 440--445.
	
	\bibitem{Rabi08}
	M.~Rabi and K.~H. Johansson, ``Event-triggered strategies for industrial
	control over wireless networks,'' in \emph{Proc. 4th Annual Int. Conf.
		Wireless Internet}.\hskip 1em plus 0.5em minus 0.4em\relax ICST (Institute
	for Computer Sciences, Social-Informatics and Telecommunications
	Engineering), 2008, p.~34.
	
	\bibitem{MazoTACWSAN}
	M.~Mazo~Jr. and P.~Tabuada, ``Decentralized event-triggered control over
	wireless sensor/actuator networks.'' \emph{IEEE Transactions on Automatic Control,
		Special issue on Wireless Sensor Actuator Networks}, vol.~56, no.~10, pp.
	2456--2461, Oct. 2011.
	
	\bibitem{Astrom}
	K.~{\AA}str\"{o}m and B.~Bernhardsson, ``Comparison of {Riemann} and {Lebesgue}
	sampling for first order stochastic systems,'' in \emph{Proceedings of the 41th {IEEE}
		Conference on Decision and Control ({CDC}'02)}, vol.~2, Dec. 2002, pp. 2011--2016.
	
	\bibitem{Tabuada}
	P.~Tabuada, ``Event-triggered real-time scheduling of stabilizing control
	tasks,'' \emph{{IEEE} Transactions on Automatic Control}, vol.~52, no.~9, pp. 1680--1685,
	Sep. 2007.
	
	\bibitem{Velasco}
	M.~Velasco, J.~Fuertes, and P.~Marti, ``{The self triggered task model for
		real-time control systems},'' in \emph{Proceedings of the 24th IEEE Real-Time Systems
		Symposium (Work in Progress)}, 2003, pp. 67--70.
	
	\bibitem{Heemels08}
	W.~Heemels, J.~Sandee, and P.~van~den Bosch, ``{Analysis of event-driven
		controllers for linear systems},'' \emph{International Journal of Control}, vol.~81,
	no.~4, pp. 571--590, 2008.
	
	\bibitem{AntaTAC}
	A.~Anta and P.~Tabuada, ``To sample or not to sample: Self-triggered control
	for nonlinear systems,'' \emph{{IEEE} Transactions on Automatic Control}, vol.~55, pp.
	2030--2042, Sep. 2010.
	
	\bibitem{MazoAnta}
	M.~Mazo~Jr., A.~Anta, and P.~Tabuada, ``{An ISS self-triggered implementation
		of linear controller},'' \emph{Automatica}, vol.~46, pp. 1310--1314, Aug.
	2010.
	
	\bibitem{sprunt1989aperiodic}
	B.~Sprunt, L.~Sha, and J.~Lehoczky, ``Aperiodic task scheduling for
	hard-real-time systems,'' \emph{Real-Time Systems}, vol.~1, no.~1, pp.
	27--60, 1989.
	
	\bibitem{buttazzo2011hard}
	G.~C. Buttazzo, \emph{Hard real-time computing systems: predictable scheduling
		algorithms and applications}.\hskip 1em plus 0.5em minus 0.4em\relax
	Springer, 2011, vol.~24.
	
	\bibitem{WalshYe2001}
	G.~C. Walsh and H.~Ye, ``{Scheduling of Networked Control Systems},''
	\emph{IEEE Control Systems Magazine}, 2001.
	
	\bibitem{Cervin08}
	A.~Cervin and T.~Henningsson, ``Scheduling of event-triggered controllers on a
	shared network,'' in \emph{Proceedings of the 47th {IEEE} Conference on Decision and Control
		({CDC}'08)}, Dec. 2008, pp. 3601--3606.
	
	\bibitem{AlAreqi2013}
	S.~Al-Areqi, D.~Gorges, S.~Reimann, and S.~Liu, ``Event-based control and
	scheduling codesign of networked embedded control systems,'' in \emph{Proceedings of the
		32nd American Control Conference ({ACC}'13)}, Jun. 2013, pp. 5299--5304.
	
	\bibitem{AlAreqi2014}
	S.~Al-Areqi, D.~Gorges, and S.~Liu, ``Stochastic event-based control and
	scheduling of large-scale networked control systems,'' in \emph{Proceedings of the
		European Control Conference}, Jun. 2014, pp. 2316--2321.
	
	\bibitem{Reimann2013}
	S.~Reimann, S.~Al-Areqi, and S.~Liu, ``An event-based online scheduling
	approach for networked embedded control systems,'' in \emph{Proceedings of the 32nd American
		Control Conference ({ACC}'13)}, Jun. 2013, pp. 5326--5331.
	
	\bibitem{Bouyer2005}
	P.~Bouyer, F.~Cassez, E.~Fleury, and K.~Larsen, ``Optimal strategies in priced
	timed game automata,'' in \emph{Foundations of Software Technology and
		Theoretical Computer Science ({FSTTCS'04})}, ser. Lecture Notes in Computer
	Science, K.~Lodaya and M.~Mahajan, Eds.\hskip 1em plus 0.5em minus
	0.4em\relax Springer, Heidelberg, 2005, vol. 3328, pp. 148--160.
	
	\bibitem{Bouyer2005a}
	------, ``Synthesis of optimal strategies using \textsc{HyTech},''
	\emph{Electronic Notes in Theoretical Computer Science}, vol. 119, no.~1, pp.
	11--31, 2005.
	
	\bibitem{Maler1995}
	O.~Maler, A.~Pnueli, and J.~Sifakis, ``\BIBforeignlanguage{English}{On the
		synthesis of discrete controllers for timed systems},'' in
	\emph{\BIBforeignlanguage{English}{Proc. 12th Symp. on Theoretical Aspects of
			Computer Science ({STACS}'95)}}, ser. Lecture Notes in Computer Science,
	E.~Mayr and C.~Puech, Eds.\hskip 1em plus 0.5em minus 0.4em\relax Springer,
	Heidelberg, 1995, vol. 900, pp. 229--242.
	
	\bibitem{ArmanTCoN}
	A.~Sharifi~Kolarijani and M.~Mazo~Jr., ``A formal traffic characterization of
	{LTI} event-triggered control systems,'' \emph{CoRR}, vol. abs/1503.05816,
	2015.
	
	\bibitem{DeAlfaro2001}
	L.~De~Alfaro, T.~Henzinger, and R.~Majumdar, ``Symbolic algorithms for
	infinite-state games,'' in \emph{Concurrency Theory ({CONCUR}'01)}, ser.
	Lecture Notes in Computer Science, K.~Larsen and M.~Nielsen, Eds.\hskip 1em
	plus 0.5em minus 0.4em\relax Springer, Heidelberg, 2001, vol. 2154, pp.
	536--550.
	
	\bibitem{Cassez2005}
	F.~Cassez, A.~David, E.~Fleury, K.~Larsen, and D.~Lime, ``Efficient on-the-fly
	algorithms for the analysis of timed games,'' in \emph{Concurrency Theory
		({CONCUR}'05)}, ser. Lecture Notes in Computer Science, M.~Abadi and
	L.~de~Alfaro, Eds.\hskip 1em plus 0.5em minus 0.4em\relax Springer,
	Heidelberg, 2005, vol. 3653, pp. 66--80.
	
	\bibitem{Asarin1998}
	E.~Asarin, O.~Maler, A.~Pnueli, and J.~Sifakis, ``Controller synthesis for
	timed automata,'' in \emph{Proc. IFAC Symp. on System Structure \& Control},
	1998, pp. 469--474.
	
	\bibitem{Alur1994}
	R.~Alur and D.~Dill, ``A theory of timed automata,'' \emph{Theoretical Computer
		Science}, vol. 126, no.~2, pp. 183--235, 1994.
	
	\bibitem{Ravn2011}
	A.~Ravn, J.~Srba, and S.~Vighio, ``Modelling and verification of web services
	business activity protocol,'' in \emph{Tools and Algorithms for the
		Construction and Analysis of Systems ({TACAS}'11)}, ser. Lecture Notes in
	Computer Science, P.~Abdulla and K.~Leino, Eds.\hskip 1em plus 0.5em minus
	0.4em\relax Springer, Heidelberg, 2011, vol. 6605, pp. 357--371.
	
	\bibitem{Havelund1997}
	K.~Havelund, A.~Skou, K.~Larsen, and K.~Lund, ``Formal modeling and analysis of
	an audio/video protocol: an industrial case study using {UPPAAL},'' in
	\emph{Proceedings of the 18th IEEE Real-Time Systems Symposium ({RTSS}'97)},
	Dec. 1997, pp. 2--13.
	
	\bibitem{DArgenio1997}
	P.~D'Argenio, J.-P. Katoen, T.~Ruys, and J.~Tretmans, ``The bounded
	retransmission protocol must be on time!'' in \emph{Tools and Algorithms for
		the Construction and Analysis of Systems ({TACAS}'97)}, ser. Lecture Notes in
	Computer Science, E.~Brinksma, Ed.\hskip 1em plus 0.5em minus 0.4em\relax
	Springer, Heidelberg, 1997, vol. 1217, pp. 416--431.
	
	\bibitem{Aceto1998}
	L.~Aceto, A.~Burgue\~{n}o, and K.~Larsen, ``\BIBforeignlanguage{English}{Model
		checking via reachability testing for timed automata},'' in
	\emph{\BIBforeignlanguage{English}{Tools and Algorithms for the Construction
			and Analysis of Systems ({TACAS}'98)}}, ser. Lecture Notes in Computer
	Science, B.~Steffen, Ed.\hskip 1em plus 0.5em minus 0.4em\relax Springer,
	Heidelberg, 1998, vol. 1384, pp. 263--280.
	
	\bibitem{Jensen1996}
	H.~Jensen, K.~Larsen, and A.~Skou, ``Modelling and analysis of a collision
	avoidance protocol using {SPIN} and {UPPAAL},'' \emph{BRICS Report Series},
	vol.~3, no.~24, 1996.
	
	\bibitem{David2000}
	A.~David and W.~Yi, ``Modelling and analysis of a commercial field bus
	protocol,'' in \emph{12th Euromicro Conference on Real-Time Systems}, 2000,
	pp. 165--172.
	
	\bibitem{Henzinger1994}
	T.~Henzinger, X.~Nicollin, J.~Sifakis, and S.~Yovine, ``Symbolic model checking
	for real-time systems,'' \emph{Information and Computation}, vol. 111, no.~2,
	pp. 193--244, 1994.
	
	\bibitem{Fehnker1999}
	A.~Fehnker, ``Scheduling a steel plant with timed automata,'' in \emph{Proc.
		6th Int. Conf. Real-Time Computing Systems and Applications ({RTCSA}'99)},
	1999, pp. 280--286.
	
	\bibitem{Abdeddaim2001}
	Y.~Abdedda\"{i}m and O.~Maler, ``Job-shop scheduling using timed automata?'' in
	\emph{Computer Aided Verification ({CAV}'01)}, ser. Lecture Notes in Computer
	Science, G.~Berry, H.~Comon, and A.~Finkel, Eds.\hskip 1em plus 0.5em minus
	0.4em\relax Springer, Heidelberg, 2001, vol. 2102, pp. 478--492.
	
	\bibitem{Abdeddaim2003}
	Y.~Abdedda\"{i}m, A.~Kerbaa, and O.~Maler, ``Task graph scheduling using timed
	automata,'' in \emph{Proc. Int. Parallel and Distributed Processing Symposium
		({IPDPS}'03)}, Apr. 2003, pp. 8 pp.--.
	
	\bibitem{Behrmann2005a}
	G.~Behrmann, E.~Brinksma, M.~Hendriks, and A.~Mader, ``Scheduling lacquer
	production by reachability analysis -- a case study,'' in \emph{Proceedings of the 16th
		IFAC World Congress}.\hskip 1em plus 0.5em minus 0.4em\relax Elsevier, 2005.
	
	\bibitem{Behrmann2005b}
	------, ``Production scheduling by reachability analysis - a case study,'' in
	\emph{Proc. 19th IEEE Int. Parallel and Distributed Processing Symposium
		({IPDPS}'05)}, Apr. 2005, pp. 140a--140a.
	
	\bibitem{Behrmann2001}
	G.~Behrmann, A.~Fehnker, T.~Hune, K.~Larsen, P.~Pettersson, J.~Romijn, and
	F.~Vaandrager, ``Minimum-cost reachability for priced time automata,'' in
	\emph{Hybrid Systems: Computation and Control ({HSCC}'01)}, ser. Lecture
	Notes in Computer Science, M.~Di~Benedetto and A.~Sangiovanni-Vincentelli,
	Eds.\hskip 1em plus 0.5em minus 0.4em\relax Springer, Heidelberg, 2001, vol.
	2034, pp. 147--161.
	
	\bibitem{Alur2001}
	R.~Alur, S.~La~Torre, and G.~Pappas, ``Optimal paths in weighted timed
	automata,'' in \emph{Hybrid Systems: Computation and Control ({HSCC}'01)},
	ser. Lecture Notes in Computer Science, M.~Di~Benedetto and
	A.~Sangiovanni-Vincentelli, Eds.\hskip 1em plus 0.5em minus 0.4em\relax
	Springer, Heidelberg, 2001, vol. 2034, pp. 49--62.
	
	\bibitem{Bouyer2008}
	P.~Bouyer, E.~Brinksma, and K.~Larsen, ``Optimal infinite scheduling for
	multi-priced timed automata,'' \emph{Formal Methods in System Design},
	vol.~32, no.~1, pp. 3--23, 2008.
	
	\bibitem{Bengtsson2004}
	J.~Bengtsson and W.~Yi, ``Timed automata: Semantics, algorithms and tools,'' in
	\emph{Lectures on Concurrency and Petri Nets}, ser. Lecture Notes in Computer
	Science, J.~Desel, W.~Reisig, and G.~Rozenberg, Eds.\hskip 1em plus 0.5em
	minus 0.4em\relax Springer, Heidelberg, 2004, vol. 3098, pp. 87--124.
	
	\bibitem{uppaal2004}
	G.~Behrmann, A.~David, and K.~Larsen, ``A tutorial on \textsc{Uppaal},'' in
	\emph{Formal Methods for the Design of Real-Time Systems ({SFM-RT'04})}, ser.
	Lecture Notes in Computer Science, M.~Bernardo and F.~Corradini, Eds., vol.
	3185.\hskip 1em plus 0.5em minus 0.4em\relax Springer, Heidelberg, Sep. 2004,
	pp. 200--236.
	
	\bibitem{Hetel2011}
	L.~Hetel, A.~Kruszewski, W.~Perruquetti, and J.-P. Richard, ``Discrete and
	intersample analysis of systems with aperiodic sampling,'' \emph{{IEEE}
		Transactions on Automatic Control}, vol.~56, no.~7, pp. 1696--1701, Jul. 2011.
	
\end{thebibliography}

\end{document}